\documentclass[12pt]{article}       

\usepackage[cmex10]{amsmath}
\usepackage{amssymb,color,graphicx,algorithm,algorithmic,amsthm}
\usepackage{array,arydshln,multirow}

\newtheorem{thm}{Theorem}
\newtheorem{lem}[thm]{Lemma}
\newtheorem{prop}[thm]{Proposition}
\newtheorem{cor}[thm]{Corollary}
\newtheorem{conj}[thm]{Conjecture}
\newtheorem{defini}[thm]{Definition}
\newenvironment{defi}{\begin{defini}\rm}{\end{defini}}

\newtheorem{ex}[thm]{Example}
\newtheorem{rema}[thm]{Remark}
\newenvironment{rem}{\begin{rema}\rm}{\end{rema}}

\newcommand{\field}[1]{\mathbb{#1}}
\newcommand{\F}{\field{F}}

\begin{document}
\author{Anna-Lena Horlemann-Trautmann\footnote{Faculty of Mathematics and Statistics, University of St.\ Gallen, Switzerland} \and Alessandro Neri\footnote{Institute of Mathematics, University of Zurich, Switzerland} }
\title{A Complete Classification of Partial-MDS (Maximally Recoverable) Codes with One Global Parity}
\maketitle

\section{Introduction}

\emph{Partial-MDS (PMDS) codes} are a family of locally repairable codes, mainly used for distributed storage. They are defined to be able to correct any pattern of $s$ additional erasures, after a given number of erasures per locality group have occurred. This makes them also \emph{maximally recoverable (MR) codes}, another class of locally repairable codes. Both terms will be properly defined in the next section.

It is known that MR codes in general, and PMDS codes in particular, exist for any set of parameters, if the field size is large enough \cite{ch07}. Moreover, some explicit constructions of PMDS codes are known, mostly (but not always) with a strong restriction on the number of erasures that can be corrected per locality group \cite{bl13,bl14,bl16,ch15,go14}. In this paper we generalize the notion of PMDS codes to allow locality groups of different sizes. 
We give a general construction of such PMDS codes with $s=1$ global parity, i.e., one additional erasure can be corrected. This construction can be seen as a generalization of the code construction from \cite{ch15}. Furthermore, we show that all PMDS codes for the given parameters are of this form, i.e., we give a classification of these codes. This implies a necessary and sufficient condition on the underlying field size for the existence of these codes (assuming that the MDS conjecture is true). For some parameter sets our generalized construction gives rise to PMDS codes with a smaller field size than any other known construction. 

The paper is structured as follows. The following section gives some preliminaries, among others the definition of maximally recoverable and PMDS codes. We show that the generalized definition of PMDS codes fulfills the MR property. Furthermore, we show when PMDS codes are MDS codes. Moreover, related results are listed. In Section \ref{sec:const} we give a general construction for PMDS codes with one global parity, by giving a generator matrix in systematic form. In Section \ref{sec:class} we give the counterpart of this result, showing that all PMDS codes with one global parity have a generator matrix of basically this form. This gives rise to a necessary and sufficient condition on the underlying field size for the existence of such codes. Section \ref{sec:decoding} deals with decoding PMDS codes with one global parity. A general decoding method and its complexity order is given, and some improvements for special types of PMDS codes are mentioned. We conclude the main work in Section \ref{sec:conclusions}. In the appendix we investigate generalizations to PMDS with more than one global parity and 
derive some necessary conditions for the existence of such codes.


\section{Preliminaries}

In a distributed storage system we store a file $x\in \F_q^k$, encoded as some codeword $c\in \F_q^n$, over several storage nodes. For simplicity we assume that each node 
stores one coordinate of $c$.  If some of these nodes fail, we want to be able to recover the lost information with as little "effort" as possible. One of the important parameters in this context is the \emph{locality} of a code for such a distributed storage system, which is the number of nodes one has to contact to repair a lost node. We call the set of nodes one has to contact if a given node fails, the locality group of that node.  The topology given by the set of all locality groups is also called a \emph{configuration}.

\begin{defi}
A code is called \emph{maximally recoverable (MR}) for a given configuration, if any erasure pattern that is information theoretically correctable is correctable. 
\end{defi}

From now on we consider a distributed storage system with $m$ disjoint locality groups, where the $i$-th group is of size $n_i$ ($i=1,\dots,m$) and can correct any $r_i$ erasures. Analogously we can separate the coordinates of the code (of length $n$) into blocks of length $n_1, n_2, \dots, n_m$ such that $\sum_{i=1}^m n_i=n$ and such that each block represents a locality group. Furthermore, we fix the locality for the whole code to be $\ell $.  

We denote an MDS code of length $n$ and dimension $k$ (and hence with minimum distance $n-k+1$) by $[n,k]$-MDS code.\footnote{We assume that the reader is familiar with the concept of MDS (maximum distance separable) codes, otherwise we refer to an introductory book on coding theory, e.g., \cite{ma77}.} We can now define PMDS codes, generalizing the definition of Blaum-Hafner-Hetzler \cite{bl13} to locality groups of different sizes but with a fixed locality $\ell$: 

\begin{defi}
 Let $\ell,m, r_1,\dots,r_m \in \mathbb N$. Define $n:=\sum_{i=1}^m (r_i+\ell)$ and let $C\subseteq \F_q^n$ be a linear code of dimension $k<n$
 with generator matrix
$$G= \left( B_1 \mid \dots \mid B_m \right) \in \F_q^{k\times n}.$$
such that $B_i\in \F_q^{k\times ( r_i + \ell)}$. 
 Then $C$ is a $[n,k, \ell ; r_1,\dots,r_m]$-\emph{partial-MDS (PMDS) code} if 
 \begin{itemize}
 \item 
 for $i\in \{1,\dots,m\}$ the row space of $B_i$ is a $[r_i+\ell, \ell ]$-MDS code, and 
 \item
for any $r_i$ erasures in the $i$-th block ($i=1,\dots, m$), the remaining code (after puncturing the coordinates of the erasures) is a $[m\ell, k]$-MDS code.
\end{itemize}
\end{defi}

The idea of PMDS codes is to be able to correct more erasures than the prescribed $r_i$ erasures per block. In particular, the erasure correction capability of PMDS code is as follows.

\begin{lem}\label{lem:PMDScap}
A $[n,k, \ell ; r_1,\dots,r_m]$-PMDS code can correct any $r_i$ erasures in the $i$-th block (simultaneously) plus $s:= m\ell -k$ additional erasures anywhere in the code.
\end{lem}
\begin{proof}
The code can correct any $r_i$ erasures in the $i$-th block, because the block represents a $[r_i+\ell, \ell ]$-MDS code. Furthermore, if we puncture the code in the $r_i$ erased coordinates in each block, by definition the remaining code is an MDS code of length $n-\sum_{i=1}^m r_i = m\ell$ and dimension at most $k$, which implies the statement.
\end{proof}

One can easily see that for the above definition to make sense we need $k\geq \ell$, which we will assume for the whole paper.
If we have equality then there exist only trivial PMDS codes, in the sense that they are MDS codes:

\begin{prop}\label{prop:triv}
If $k= \ell$, then a code is a $[n,k, k; r_1,\dots,r_m]$-PMDS code if and only if it is a $[n,k]$-MDS code.
\end{prop}
\begin{proof}
Assume the code is $[n,k, k; r_1,\dots,r_m]$-PMDS. 
Let $S$ be a submatrix of the generator matrix $G= \left( B_1 \mid \dots \mid B_m \right) $ after erasing $r_i$ columns per block $B_i$. Then each block $B_i$ still has $k$ columns, i.e., for $S$ to generate an MDS code all minors (including minors completely inside one block) need to be non-zero. Since we need to check the MDS property for any such $S$, all minors of $G$ need to be non-zero, i.e., $G$ generates an MDS code.

For the other direction assume that $C$ is a $[n,k]$-MDS code. Then any punctured version of $C$ is also an MDS code of dimension $k=\ell$. In particular, the two conditions for the PDMS property are fulfilled.
\end{proof}

We note that Proposition \ref{prop:triv} also includes the case $m=1$, since this automatically implies $k=\ell$. 
Furthermore, the case $k=1$ is also included, since this implies $\ell=k=1$.

\begin{rem}
In \cite{bl13} PMDS codes are studied with respect to RAID architectures, where all blocks have the same size $n_1=\ell+r_1$, such that a codeword is written in an $m\times n_1$ array and complete columns are erased when a RAID disk fails. Moreover, due to the physical nature of solid state disks, with age $s$ additional erasures may occur anywhere in the codeword.
\end{rem}

It was shown in \cite[Lemma 4]{go14} that, in the case that the locality groups are disjoint and $r_i=1$ for $i=1,\dots,m$, the MR property is equivalent to the PMDS property. That MR implies PMDS for configurations with disjoint locality blocks is straight-forward. The other direction was proved by showing that any other erasure pattern than at most one erasures per block plus $s$ extra erasures anywhere cannot be correctable at all. Thus any $[n,k,\ell; 1,\dots,1]$-PMDS code is also an MR code. 
We will now generalize this result to variable values of $r_1,\dots,r_m$.

\begin{thm}
A $[n,k, \ell; r_1,\dots,r_m]$-PMDS code is maximally recoverable.
\end{thm}
\begin{proof}
We know by Lemma \ref{lem:PMDScap} that any erasure pattern with at most $r_i$ erasures in the $i$-th block, plus $s=m\ell -k$ additional erasures anywhere in the code is correctable with a $[n,k, \ell; r_1,\dots,r_m]$-PMDS code. It remains to show that any other erasure pattern is not correctable by any code of length $n$ and dimension $k$ for the given locality conditions.

Let $E$ be an erasure pattern affecting $m'<m$ locality groups (WLOG we assume that the first $m'$ groups are affected) with $|E|>\sum_{i=1}^{m'} r_i +s$. 
Extend $E$ to a larger pattern of erasures $E'$ erasing $r_i$ arbitrary coordinates in the $i$-th locality group for $i=m'+1,\dots,m$. It holds that $E'$ is correctable if $E$ is correctable since the $i$-th locality group can correct any $r_i$ erasures. Note that the size of $E'$ exceeds the redundancy of the code, since
$$|E'| > \sum_{i=1}^{m} r_i +s = n-m\ell +s =n-k ,$$ 
hence $E'$ is not correctable by any code of length $n$ and dimension $k$.

Similarly, if all $m$ locality groups are affected by an erasure pattern $E$ and we have $|E|>\sum_{i=1}^{m} r_i +s$, then $|E| >n-k $, and hence $E$ is not correctable.
\end{proof}

In the following we give a brief overview of known results for (non-trivial) PMDS codes.

\begin{prop}\cite{ch07}
MR codes of length $n$ and dimension $k$ exist for any configuration over any finite field of size $q> \binom{n-1}{k-1}$.
\end{prop}

Since MR codes are PMDS codes for disjoint locality blocks, the above result also implies that PMDS codes exist for any set of parameters if the field size is large enough.

In the following results it is assumed that all locality blocks have the same length $n_i=n/m$  (for $i=1,\dots,m$), which is why we must assume that $m| n$. 
A general construction of PMDS codes, based on rank-metric and MDS codes, was given in \cite{ca17}. This gives the following existence result:
\begin{prop}\cite{ca17}
$[n,k,\ell; r,\dots,r]$-PMDS codes with $m$ locality blocks of the same length exist over a finite field of size  $q^{n-mr}$.
\end{prop}

Furthermore, some specific constructions of PMDS codes, either for small $r$ or small $s$, are given in 
\cite{bl13,bl14,bl16,go14}. 
In \cite{bl13} a construction of PMDS codes with $s=1$ and equal block length $n_i=n/m$ over $\F_q$ with $q=2^b\geq \max\{n_i,m\}$ was given. 

Another construction for PMDS codes with $s=1$ and equal block length $n_i=n/m$, requiring field size $q\geq n_i$, was given in \cite[Theorem 1]{ch15}:
\begin{prop}\cite{ch15}
$[n,m\ell-1,\ell; r,\dots,r]$-PMDS codes exist over any finite field of size $q\geq n/m$.
\end{prop}
This construction (as the ones of \cite{bl13}) is based on Vandermonde matrices and thus equivalent to using generalized Reed-Solomon codes as building blocks. In the following section we will give a generalized construction of PMDS codes with $s=1$, allowing various block lengths and any MDS codes as building blocks. In some cases this generalized construction will allow us to reduce the field size compared to the construction of \cite{ch15}. Moreover, in contrast to the construction of \cite{ch15}, our construction provides generator matrices in systematic (or standard) form.

Note that a natural lower bound on the field size is given by the condition that every block constitutes a $[r_i + \ell, \ell]$-MDS code. To derive a bound from this condition we assume that the MDS conjecture (\cite{se55}, see also \cite[Conjecture 11.16]{roth06}) is true (and it has been proven for many parameter sets):

\begin{conj}[MDS Conjecture]\label{MDSconj}\cite{se55}
A $[n,k]$-MDS code with $1<k<n-1$ over $\F_q$ has length $n\leq q+1$, unless $q=2^h$ and $k\in \{3,q-1\}$, in which case $n\leq q+2$.
\end{conj}
It follows that, for $\ell >1$, a $[n,k,\ell; r_1,\dots,r_m]$-PMDS code cannot exist over $\F_q$ if $q<\max_i\{r_i + \ell -1\}$, except if $(\ell, \max_i\{r_i\}) \in \{(3,2^h-1), (2^h-1,3)\}$, in which case such a code cannot exist if  $q<\max_i\{r_i + \ell -2\}$.
In Corollary \ref{co:MRr2} of this paper we will show that, in the case that $s = 1$, this bound cannot be obtained in general, but that it has to be increased by $1$. Furthermore, we show in Corollary \ref{cor:23} that, for $s>1$, a lower bound for the field size is given by $\max_i\{r_i + \ell  + s-1\}$ (except for some special parameter sets).


\section{Generalized Construction of PMDS Codes with $s=1$}\label{sec:const}

We will now present a general construction for PMDS codes of dimension $k=m\ell -1$; those codes can correct $s=m\ell -k=1$ additional erasure anywhere in the code. Because of Proposition \ref{prop:triv} we assume that the number of locality blocks is $m>1$.

\begin{thm}\label{th:MRri}
For any integers $m\geq 2 $ and $\ell,r_1,\dots,r_m \geq 1$, the following generator matrix $G$ generates a $[n,k, \ell; r_1,\dots,r_m]$-PMDS code of length $n=  m\ell+ \sum_i r_i$, dimension $k=m\ell -1$ and locality $\ell$  over $\F_q$:
$$ G= \left(\begin{array}{cccc|c} B_1 & 0 & \dots & 0 & M \\  0 & B_2 & \dots & 0 & M \\ 
\vdots & \vdots & \ddots & \vdots & \vdots \\ 0 & 0 & \dots & B_{m-1} & M \\\hline  0 & 0 & \dots & 0& A  \end{array}\right)   ,$$
where
$$ B_i=\left( \begin{array}{ccccccc} 1&0&\dots &0& x^{(i)}_{1,1}&\dots&x^{(i)}_{1,r_i} \\ 0&1&\dots&0& x^{(i)}_{2,1}&\dots&x^{(i)}_{2,r_i}\\ \vdots&&\ddots&&\vdots &  & \vdots\\ 0&0&\dots&1& x^{(i)}_{\ell,1}&\dots&x^{(i)}_{\ell,r_i} \\   \end{array}\right)       \quad   \in \F_q^{\ell\times (\ell+r_i)}  ,$$
$$ A=\left( \begin{array}{cccccccc} 1&0&\dots &0&1&  x^{(m)}_{1,1}&\dots&x^{(m)}_{1,r_m} \\ 0&1&\dots&0&1&  x^{(m)}_{2,1}&\dots&x^{(m)}_{2,r_m} \\ \vdots&&\ddots&&\vdots&\vdots & & \vdots \\ 0&0&\dots&1&1&  x^{(m)}_{\ell-1,1}&\dots&x^{(m)}_{\ell-1,r_m}   \end{array}\right)       \quad   \in \F_q^{(\ell-1)\times (\ell+r_m)}  ,$$
$$ M=\left( \begin{array}{cccccccc} 0&0&\dots &0&1& x^{(m)}_{\ell,1} &\dots&x^{(m)}_{\ell,r_m} \\ 0&0&\dots &0&1& x^{(m)}_{\ell,1}&\dots&x^{(m)}_{\ell,r_m}  \\ \vdots&&\ddots&&\vdots&\vdots & &\vdots \\  0&0&\dots &0&1& x^{(m)}_{\ell,1} &\dots&x^{(m)}_{\ell,r_m}  \end{array}\right)       \quad   \in \F_q^{\ell\times (\ell+r_m)}   ,$$
such that, for $i=1,\ldots, m-1$, the matrices
$$ \widehat{B}_i=\left(\quad B_i \quad \begin{array}{|c} 1 \\ 1 \\ \vdots \\ 1  \end{array}\right) $$
are generator matrices of a $[\ell+r_i+1, \ell]$-MDS code and
$$ \widehat{A}=\left(\begin{array}{c|ccccccc} 1 & 0 & \dots & 0 & 1& x^{(m)}_{\ell,1} &\dots&x^{(m)}_{\ell,r_m}   \\ \hline 0  & & \\ \vdots & & & & A &  & \\  0 \end{array}\right)$$
is a generator matrix of a $[\ell+r_m+1, \ell]$-MDS code. 

The PMDS code can correct any $r_i$ erasures in the $i$-th block (simultaneously) plus $s=1$ additional erasure anywhere.
\end{thm}

\begin{proof}
It is easy to see that the row space of each block is $\ell$-dimensional. For any $i \in \{1,\ldots, m-1 \}$ we have by construction that $\widehat{B}_i$ generates an MDS code, and therefore $B_i$ is the generator matrix of a $[\ell+r_i,\ell]$-MDS code. This implies the locality for the first $m-1$ blocks. Moreover, also the matrix 
 $ \widehat{A}$ without the first column generates a $[\ell+r_i,\ell]$-MDS code, and this implies the locality for the last block.

We now check that after puncturing in $r_i$ coordinates per block $B_i$ the code is MDS. Recall that $k=m\ell-1$. For the maximal minors of the punctured generator matrix of size $k\times k$ we only have the choice of choosing $\ell$ columns from all but one blocks and choosing $\ell-1$ columns from the remaining block. Whenever we choose $\ell$ columns from a block, it suffices to compute all minors with only one choice of $\ell$ out of $\binom{\ell+r_i}{\ell}$ columns per block, since any choice of $\ell$ columns is a basis for the column space of $B_i$. For simplicity we always choose the first columns of a block. We distinguish two cases:

\underline{Case 1 ($\ell-1$ columns from last block):} The columns of the first $(m-1)$ blocks form an upper left $\ell(m-1)$-identity submatrix, hence the big minor is determined by the last $\ell-1$ columns, i.e., the maximal minors of $A$. These minors correspond to the maximal minors of $\widehat{A}$ containing the first column. Since $\widehat{A}$ generates an MDS code they are all non-zero.

\underline{Case 2 ($\ell-1$ columns from one of the first $m-1$ blocks): } We consider taking $\ell-1$ columns from the first block. The other blocks work analogously. After row and column operations the corresponding minors are of the form 
$$\pm \det\left(\begin{array}{c|cccc} I_{m\ell-1-t} & & & * \\\hline  &1& x^{(1)}_{j_1,h_1} & \dots & x^{(1)}_{j_1,h_{t-1}}   \\ 0 & \vdots & \vdots & & \vdots \\ &1& x^{(1)}_{j_t,h_1} & \dots & x^{(1)}_{j_t,h_{t-1}} \end{array}\right) $$ 
$$= \pm \det\left(\begin{array}{cccc}  x^{(1)}_{j_1,h_1} & \dots & x^{(1)}_{j_1,h_{t-1}}&1   \\  \vdots&& \vdots  & \vdots \\  x^{(1)}_{j_t,h_1} & \dots & x^{(1)}_{j_t,h_{t-1}} &1 \end{array}\right) , $$ 
for some 
$1\leq t \leq \ell$, $1\leq j_1 <j_2<\ldots <j_t\leq \ell$ and $1\leq h_1<h_2<\ldots<h_{t-1}\leq r_1$. They are all non-zero since the matrix $\widehat{B}_1$ is the generator matrix of an MDS code.
\end{proof}

\begin{rem}
If $G$ is of the form described in Theorem \ref{th:MRri},
then the conditions that $\widehat{B}_i$ and $\widehat{A}$ are generator matrices of MDS codes are also necessary for the generated code to be $[n,k,\ell; r_1,\dots,r_m]$-PMDS. We will show a more general result in Theorem \ref{th:classMRri}.
\end{rem}

To finalize the construction of PMDS codes with $s=1$ we need to show that we can always find the matrices $B_i, A, M$ from Theorem \ref{th:MRri}. For this we need the following lemma:

\begin{lem}\label{lem:equiv}
Let $\ell, r \in \mathbb N$, and $\alpha_1,\ldots, \alpha_{\ell}\in \F_q^*$. Then the matrix
$$G=\left(\quad I_{\ell}\quad\begin{array}{|cccc} x_{1,1} & x_{1,2} & \dots & x_{1,r}\\x_{2,1} & x_{2,2} & \dots & x_{2,r}\\ \vdots & \vdots & & \vdots  \\x_{\ell,1} & x_{\ell,2} & \dots & x_{\ell,r}
\end{array}\right)$$
generates a $[\ell+r, \ell]$-MDS code if and only if the matrix
$$\widehat{G}=\left(\quad I_{\ell} \quad\begin{array}{|cccc} \alpha_1x_{1,1} & \alpha_1x_{1,2} & \dots &\alpha_1 x_{1,r}\\\alpha_2x_{2,1} & \alpha_2x_{2,2} & \dots & \alpha_2x_{2,r}\\ \vdots & \vdots & & \vdots  \\\alpha_{\ell}x_{\ell,1} & \alpha_{\ell}x_{\ell,2} & \dots & \alpha_{\ell}x_{\ell,r} \end{array}\right) $$
generates  a $[\ell+r, \ell]$-MDS code.
\end{lem}
\begin{proof}
It was shown in \cite{ro85} that $G$ generates an MDS code if and only if 
$$X=\left(\begin{array}{cccc} x_{1,1} & x_{1,2} & \dots & x_{1,r}\\x_{2,1} & x_{2,2} & \dots & x_{2,r}\\ \vdots & \vdots & & \vdots  \\x_{\ell,1} & x_{\ell,2} & \dots & x_{\ell,r}\end{array}\right)$$
is superregular, i.e., if all minors of $X$ are non-zero. For any square submatrix of $X$ we have that
$$\det\left(\begin{array}{cccc} x_{i_1,j_1} & x_{i_1,j_2} & \dots & x_{i_1,j_t}\\x_{i_2,j_1} & x_{i_2,j_2} & \dots & x_{i_2,j_t}\\ \vdots & \vdots & & \vdots  \\x_{i_t,j_1} & x_{i_t,j_2} & \dots & x_{i_t,j_t}\end{array}\right) = 0 \iff $$
$$\det\left( \mathrm{diag}(\alpha_{i_1},\ldots, \alpha_{i_t})   \left(\begin{array}{cccc} x_{i_1,j_1} & x_{i_1,j_2} & \dots & x_{i_1,j_t}\\x_{i_2,j_1} & x_{i_2,j_2} & \dots & x_{i_2,j_t}\\ \vdots & \vdots & & \vdots  \\x_{i_t,j_1} & x_{i_t,j_2} & \dots & x_{i_t,j_t}\end{array}\right)  \right) =0$$
$$\iff\det\left(\begin{array}{cccc} \alpha_{i_1}x_{i_1,j_1} & \alpha_{i_1}x_{i_1,j_2} & \dots & \alpha_{i_1}x_{i_1,j_t}\\\alpha_{i_2}x_{i_2,j_1} & \alpha_{i_2}x_{i_2,j_2} & \dots & \alpha_{i_2}x_{i_2,j_t}\\ \vdots & \vdots & & \vdots  \\ \alpha_{i_t}x_{i_t,j_1} & \alpha_{i_t}x_{i_t,j_2} & \dots & \alpha_{i_t}x_{i_t,j_t}\end{array}\right) = 0 ,$$
which implies the statement.
\end{proof}

\begin{cor}\label{co:MRr}
\begin{enumerate}
\item
For any integers $m\geq 2 $ and $\ell,r_1,\dots,r_m \geq 1$ there exists a $[n,k=m\ell-1, \ell; r_1,\dots,r_m]$-PMDS code over any field $\F_q$ with $q\geq \max_i\{r_i\}+\ell$.

\item
If there exists $h\in \mathbb N$ such that $\ell\in \{3,2^h-1\}$ and $\max_i \{r_i\} +\ell = 2^h+1$, then there exists an  $[n,k=m\ell-1, \ell; r_1,\dots,r_m]$-PMDS code over $\F_q$ with $q=2^h= \max_i\{r_i\}+\ell -1$.

\item
If $\ell=1$, then there exists a $[n,k=m-1, 1; r_1,\dots,r_m]$-PMDS code over $\F_q$ with $q\geq 2$.
\end{enumerate}
\end{cor}

\begin{proof}
\begin{enumerate}
\item
Let $q\geq \max_i\{r_i\} +\ell$. Then we know that $[\ell +r_i+1,\ell ]$-MDS codes exist over $\F_q$, namely extended Reed-Solomon codes. Furthermore, we know that each $[\ell +r_i+1,\ell ]$-MDS code has a generator matrix in systematic form $(I_\ell \mid N)$, where $N\in (\F_q^*)^{\ell\times (r_i+1)}$. 
To prove the statement of the corollary, by using Theorem \ref{th:MRri}, we need to show that there exist $[\ell +r_i+1,\ell ]$-MDS codes over $\F_q$ with generator matrices of the form 
$$\left( \begin{array}{c|c|c} \hspace{1cm} \;& \hspace{1cm} \;&1 \\ &&1 \\ I_\ell &* & \vdots \\ &&1  \\ &&1 \end{array}\right) \textnormal{ and } \left(\begin{array}{c|ccc|c|c} 1 & 0 &\dots& 0 & 1   & *\\ \hline 0  & & &&& \\ \vdots &  & I_{\ell-1} &  & * &  *\\  0&&&&& \end{array}\right)$$
respectively. For both cases we can use Lemma \ref{lem:equiv} to transform any generator matrix of an extended Reed-Solomon code into the desired from.
Therefore we have shown that matrices $B_1,\dots,B_{m-1}, A, M$ fulfilling the conditions of Theorem \ref{th:MRri}, exist.

\item
For the second part we use the fact that $[n,k]$-MDS codes exist over a field of size $q=n-2$ if $n=2^h$ (for some $h\in \mathbb N$) and $k\in\{3,q-1\}$, namely doubly extended Reed-Solomon codes. Analogously to part 1.\ of this proof, we can then construct the required matrices fulfilling the conditions of Theorem \ref{th:MRri}. For this note that the matrices for the locality blocks, that are not of maximal length, can still be constructed from extended Reed-Solomon codes.

\item
For the third part note that
$$G= \left(\begin{array}{cccc|cccc|c|cccc|cccc} 1 & 1 &\dots& 1 &    0& 0&\dots&0 & \dots&   0& 0&\dots&0 &     1&1&\dots&1 \\  0  &0 &\dots &0&   1&1&\dots&1&  \dots&   0& 0&\dots&0 &    1&1&\dots&1 \\ \vdots &  &  & &&&&&\ddots &  &&&&&&&\vdots  \\  0&0&\dots&0&   0&0&\dots&0 &   \dots &    1&1&\dots&1&  1&1&\dots&1 \end{array}\right)$$
generates a $[n,m-1, 1; r_1,\dots,r_m]$-PMDS code over any finite field. 
\end{enumerate}
\end{proof}

\begin{ex}
The matrix
$$G= \left( \begin{array}{ccc|ccc} 1&0&1   & 0&1&1\\ 0&1&2& 0&1&1 \\\hline 0&0&0  &1&1&2   \end{array}  \right) $$
is a generator matrix (in systematic form) of a PMDS code with parameters $[n,k,\ell; r_1,r_2] = [6,3,2; 1,1]$  over $\F_3$.
One can easily check that the row space of each of the two blocks is a $[3,2]$-MDS code, and that each combination of two columns from each block forms a $[4,3]$-MDS code, i.e., it can correct $s=1$ additional erasure.
\end{ex}


\begin{ex}
Let $\F_4=\{0,1,\alpha,\alpha+1\}$. The matrix
$$G= \left( \begin{array}{ccccc|cccc} 1&0&0   &1&1    &0&0&1&1  \\ 0&1&0&  \alpha+1&\alpha    &0&0&1&1\\ 0&0&1&  \alpha&\alpha+1   &0&0&1&1 \\\hline 0&0&0&0&0  &1&0&1&\alpha \\  0&0&0&0&0  &0&1&1&\alpha+1  \end{array}  \right) $$
is a generator matrix (in systematic form) of a PMDS code with parameters $[n,k,\ell; r_1,r_2] = [9,5,3; 2,1]$  over $\F_4$.
One can easily check that the row space of the first block is a $[5,3]$-MDS code,  the row space of the second block is a $[4,3]$-MDS code, and that each combination of three columns from each block forms a $[6,5]$-MDS code, i.e., it can correct $s=1$ additional erasure.
\end{ex}


\section{Classification of all PMDS Codes with $s=1$}\label{sec:class}

In this section we give a complete classification of PMDS codes that can correct one additional erasure anywhere in the code, by determining the systematic form of their generator matrix. The main result in Theorem \ref{th:classMRri} also generalizes the construction of PMDS codes given in the previous section. 

\begin{lem}\label{lem:Stand}
Let $m\geq 2 $ and $\ell,r_1,\dots,r_m \geq 1$ and let $C$ be a $[n,k=m\ell-1, \ell; r_1,\dots,r_m]$-PMDS code over a field $\F_q$. Then $C$  has a generator matrix of the form
\begin{equation}\label{eq:prel}
 G= \left(\begin{array}{cccc|c} B_1 & 0 & \dots & 0 & M_1 \\  0 & B_2 & \dots & 0 & M_2 \\ 
\vdots & \vdots & \ddots & \vdots & \vdots \\ 0 & 0 & \dots & B_{m-1} & M_{m-1} \\\hline  0 & 0 & \dots & 0& A  \end{array}\right)
\end{equation}
where
\begin{equation}\label{eq:Bi}
 B_i=\left( \begin{array}{ccccccc} 1&0&\dots &0& x^{(i)}_{1,1}&\dots&x^{(i)}_{1,r_i} \\ 0&1&\dots&0& x^{(i)}_{2,1}&\dots&x^{(i)}_{2,r_i}\\ \vdots&&\ddots&&\vdots &  & \vdots\\ 0&0&\dots&1& x^{(i)}_{\ell,1}&\dots&x^{(i)}_{\ell,r_i} \\   \end{array}\right)       \quad   \in \F_q^{\ell\times (\ell+r_i)}    ,
\end{equation}
$$ A=\left( \begin{array}{cccccccc} 1&0&\dots &0&\alpha^{(m)}_1&  \alpha^{(m)}_1x^{(m)}_{1,1}&\dots&\alpha^{(m)}_1x^{(m)}_{1,r_m} \\ 0&1&\dots&0&\alpha^{(m)}_2&  \alpha^{(m)}_2 x^{(m)}_{2,1}&\dots&\alpha^{(m)}_2 x^{(m)}_{2,r_m} \\ \vdots&&\ddots&&\vdots&\vdots & & \vdots \\ 0&0&\dots&1&\alpha^{(m)}_{\ell-1}&  \alpha^{(m)}_{\ell-1}x^{(m)}_{\ell-1,1}&\dots&\alpha^{(m)}_{\ell-1}x^{(m)}_{\ell-1,r_m}   \end{array}\right)       \quad   \in \F_q^{(\ell-1)\times (\ell+r_m)}  ,$$
$$ M_i=\left( \begin{array}{cccccccc} 0&0&\dots &0&\alpha^{(i)}_{1}& \alpha^{(i)}_{1}x^{(m)}_{\ell,1} &\dots&\alpha^{(i)}_{1}x^{(m)}_{\ell,r_m} \\ 0&0&\dots &0&\alpha^{(i)}_2& \alpha^{(i)}_{2}x^{(m)}_{\ell,1}&\dots&\alpha^{(i)}_{2}x^{(m)}_{\ell,r_m}  \\ \vdots&&\ddots&&\vdots&\vdots & &\vdots \\  0&0&\dots &0&\alpha^{(i)}_{\ell}&\alpha^{(i)}_{\ell} x^{(m)}_{\ell,1} &\dots&\alpha^{(i)}_{\ell}x^{(m)}_{\ell,r_m}  \end{array}\right)       \quad   \in \F_q^{\ell\times (\ell+r_m)},$$
up to permutation of variables.
\end{lem}

\begin{proof}
 Let $C$ be a $[n,k=m\ell-1, \ell; r_1,\dots,r_m]$-PMDS code with generator matrix 
$$\widetilde{G}= \left(\widetilde{B}_1 \mid \dots \mid \widetilde{B}_m \right), $$
where $\widetilde{B}_i=(\widetilde{C}_i \mid \widetilde{D}_i)$, $\widetilde{C}_i \in \F_q^{(m\ell -1) \times \ell}$ and $\widetilde{D}_i \in \F_q^{(m\ell -1)\times r_i}$ for $i=1,\ldots, m$. By definition of PMDS codes, the code generated by $\widetilde{G}_C=\left(\widetilde{C}_1 \mid \dots \mid \widetilde{C}_m \right)$ is a $[m\ell , m\ell -1]$-MDS code, since it is obtained by puncturing the code $C$ in the coordinates defined by the blocks $\widetilde{D}_i$. Therefore there exists a matrix $S \in \mathrm{GL}_{m\ell -1}(\F_q)$ such that 
\begin{equation}\label{eq:GC}S\widetilde{G}_C=\left( I_{m\ell -1}\begin{array}{|c} \alpha_1\\ \vdots \\ \alpha_{m\ell -1} \end{array}\right)\in \F_q^{(m\ell -1)\times m\ell }, \end{equation}
where $\alpha_i\neq 0$ for every $i=1,\ldots,m\ell -1$.
Hence the generator matrix  $G:=S\widetilde{G}$ will be of the form
$$ G= \left(\begin{array}{cc|cc|c|cc|cc} 
I_{\ell} & &  0 & & \dots & 0 &  & \\ 
 0 & & I_{\ell} & &\dots & 0 & & \\ 
\vdots & S\widetilde{D}_1 & \vdots &S\widetilde{D}_1 &  & \vdots & S\widetilde{D}_{m-1}& S\widetilde{C}_{m} & S\widetilde{D}_{m} \\
 0 &  & 0 & &\dots & I_{\ell} & &  \\  
0 & & 0 & & \dots & 0& &  \end{array}\right)$$
Denote by $\mathrm{rs}(B_i)$ and $\mathrm{cs}(B_i)$ the row space and the column space of the matrix $B_i$, respectively.  
By definition of PMDS codes, we have that, for $i=1,\ldots, m-1$,
$$\ell=\dim\mathrm{rs}(\widetilde{B}_i)=\dim\mathrm{cs}(\widetilde{B}_i)=\dim\mathrm{cs}(S\widetilde{B}_i),$$
Hence $\mathrm{cs}(S\widetilde{D}_i)\subseteq \mathrm{cs}(S\widetilde{C}_i)$, and this implies that the matrix 
$S\widetilde{G}_C$ is of the form
$$\left(\begin{array}{cccc|cc} B_1 & 0 & \dots & 0 &  \\  0 & B_2 & \dots & 0 &  \\ 
\vdots & \vdots & \ddots & \vdots & S\widetilde{C}_m & S\widetilde{D}_m \\ 0 & 0 & \dots & B_{m-1} &  \\  0 & 0 & \dots & 0&  \end{array}\right),$$
where every block $B_i$ is of the form (\ref{eq:Bi}).

It remains to show that the last block is of the desired form.
By (\ref{eq:GC}), the last block is of the form
$$\left(S\widetilde{C}_m \mid S\widetilde{D}_m\right)=\left(\begin{array}{cccc|cccc}
0 & \dots & 0 & \alpha_1 & \multicolumn{3}{c}{\multirow{4}{*}{\qquad \LARGE{$X_1$}\quad}} \\
0 & \dots & 0 & \alpha_2 \\
\vdots & &\vdots &  \vdots  \\
0&\dots &0 & \alpha_{\ell(m-1)-1} \\ \hline
0 & \dots & 0 & \alpha_{\ell(m-1)} &\multicolumn{3}{c}{\multirow{4}{*}{\qquad \LARGE{$X_2$}\quad}} \\ 
\multicolumn{3}{c}{\multirow{3}{*}{\LARGE{$I_{\ell-1}$}}}& \alpha_{\ell(m-1)+1}  &  \\
 & & &  \vdots  \\
& & & \alpha_{m\ell-1} \\
 \end{array} \right)=:\Large{\left(\begin{array}{c|c}Y_1 & X_1 \\ \hline Y_2 & X_2 \end{array}\right)}.$$
Since $C$ is PMDS we get
$$\ell=\dim \mathrm{rs}(S\widetilde{B}_m)=\dim \mathrm{rs}(S\widetilde{C}_m\mid S\widetilde{D}_m )=\dim \left(\mathrm{rs}(Y_1 \mid X_1)+\mathrm{rs}(Y_2 \mid X_2)\right).$$
We observe that $\det(Y_2)=\alpha_{\ell(m-1)}\neq 0$, hence $\dim\mathrm{rs}(Y_2 \mid X_2)=\ell$, which implies 
\begin{equation*}\label{eq:rs}\mathrm{rs}(Y_1 \mid X_1)\subseteq \mathrm{rs}(Y_2 \mid X_2). \end{equation*} 
This implies, by the structure of the matrix $S\widetilde{B}_m$, that every row of $(Y_1\mid X_1)$ is a multiple of the first row of $(Y_2\mid X_2)$.
Since $X_2$ is arbitrary and $\alpha_i\neq 0$ for every $i$, we can write
$$X_2=\left(\begin{array}{cccc} \alpha_{\ell(m-1)}x^{(m)}_{\ell,1}&\dots&\alpha_{\ell(m-1)}x^{(m)}_{\ell,r_m} \\   \alpha_{\ell(m-1)+1} x^{(m)}_{1,1}&\dots&\alpha_{\ell(m-1)+1} x^{(m)}_{1,r_m} \\ \vdots & & \vdots \\  \alpha_{m\ell-1}x^{(m)}_{\ell-1,1}&\dots&\alpha_{m\ell-1}x^{(m)}_{\ell-1,r_m}   \end{array}\right),  $$
for some $x_{i,j}^{(m)} \in \F_q$. Therefore, also the last block is of the desired form.
\end{proof}

We can finally give a characterization of PMDS codes with one global parity (i.e., $s=1$):

\begin{thm}\label{th:classMRri}
For any $m\geq 2 $ and $ \ell, r_1,\dots,r_m \geq 1$, a linear code over $\F_q$ of length $n=m\ell + \sum_{i=1}^m r_i$ and dimension $k=m\ell -1$ 
 is a $[n,k,\ell; r_1,\dots,r_m]$-PMDS code \emph{if and only if} it has a generator matrix of the form
\begin{equation}\label{eq:Stand} 
G= \left(\begin{array}{cccc|c} B_1 & 0 & \dots & 0 & M_1 \\  0 & B_2 & \dots & 0 & M_2 \\ 
\vdots & \vdots & \ddots & \vdots & \vdots \\ 0 & 0 & \dots & B_{m-1} & M_{m-1} \\\hline  0 & 0 & \dots & 0& A  \end{array}\right)\end{equation}
where
$$ B_i=\left( \begin{array}{ccccccc} 1&0&\dots &0& x^{(i)}_{1,1}&\dots&x^{(i)}_{1,r_i} \\ 0&1&\dots&0& x^{(i)}_{2,1}&\dots&x^{(i)}_{2,r_i}\\ \vdots&&\ddots&&\vdots &  & \vdots\\ 0&0&\dots&1& x^{(i)}_{\ell,1}&\dots&x^{(i)}_{\ell,r_i} \\   \end{array}\right)       \quad   \in \F_q^{\ell\times (\ell+r_i)}   ,$$
$$ A=\left( \begin{array}{cccccccc} 1&0&\dots &0&\alpha^{(m)}_1&  \alpha^{(m)}_1x^{(m)}_{1,1}&\dots&\alpha^{(m)}_1x^{(m)}_{1,r_m} \\ 0&1&\dots&0&\alpha^{(m)}_2&  \alpha^{(m)}_2 x^{(m)}_{2,1}&\dots&\alpha^{(m)}_2 x^{(m)}_{2,r_m} \\ \vdots&&\ddots&&\vdots&\vdots & & \vdots \\ 0&0&\dots&1&\alpha^{(m)}_{\ell-1}&  \alpha^{(m)}_{\ell-1}x^{(m)}_{\ell-1,1}&\dots&\alpha^{(m)}_{\ell-1}x^{(m)}_{\ell-1,r_m}   \end{array}\right)       \quad   \in \F_q^{(\ell-1)\times (\ell+r_m)}    ,$$
$$ M_i=\left( \begin{array}{cccccccc} 0&0&\dots &0&\alpha^{(i)}_{1}& \alpha^{(i)}_{1}x^{(m)}_{\ell,1} &\dots&\alpha^{(i)}_{1}x^{(m)}_{\ell,r_m} \\ 0&0&\dots &0&\alpha^{(i)}_2& \alpha^{(i)}_{2}x^{(m)}_{\ell,1}&\dots&\alpha^{(i)}_{2}x^{(m)}_{\ell,r_m}  \\ \vdots&&\ddots&&\vdots&\vdots & &\vdots \\  0&0&\dots &0&\alpha^{(i)}_{\ell}&\alpha^{(i)}_{\ell} x^{(m)}_{\ell,1} &\dots&\alpha^{(i)}_{\ell}x^{(m)}_{\ell,r_m}  \end{array}\right)       \quad   \in \F_q^{\ell\times (\ell+r_m)}   , $$
such that $\alpha_j^{(i)} \neq 0$ for any $i,j$, and, for $i=1,\ldots, m-1$, the matrices
$$ \widehat{B}_i=\left(\quad B_i \quad \begin{array}{|c} \alpha^{(i)}_1 \\ \alpha^{(i)}_2 \\ \vdots \\ \alpha^{(i)}_{\ell}  \end{array}\right) ,$$
are generator matrices of a $[\ell+r_i+1, \ell]$-MDS code and
$$ \widehat{A}=\left( \begin{array}{cccccccc} 1&0&\dots &0&1&  x^{(m)}_{\ell,1}&\dots&x^{(m)}_{\ell,r_m} \\ 0&1&\dots&0&1&  x^{(m)}_{1,1}&\dots&x^{(m)}_{1,r_m} \\ \vdots&&\ddots&&\vdots&\vdots & & \vdots \\ 0&0&\dots&1&1&  x^{(m)}_{\ell-1,1}&\dots&x^{(m)}_{\ell-1,r_m}   \end{array}\right) $$
is a generator matrix of a $[\ell+r_m+1, \ell]$-MDS code.

\end{thm}

\begin{proof}
For the \emph{only if}-direction let $C$ be a $[n,k,\ell; r_1,\dots,r_m]$-PMDS code. By Lemma \ref{lem:Stand}, $C$ has a generator matrix of the form \eqref{eq:prel}. 
Since $C$ is PMDS we must have that $B_i$ generates a $[\ell + r_i, \ell]$-MDS code. Moreover, any $(m\ell-1)$-minor of $G$ with $\ell$ columns from all but one block and $\ell -1$ columns from the remaining block must be non-zero. In particular, the  matrix 
$$ \left(\begin{array}{cccc|cc} \bar B_1 & 0 & \dots & 0 & 0& \bar\alpha^{(1)} \\  0 & \bar B_2 & \dots & 0 &  0& \bar\alpha^{(2)}\\ 
\vdots & \vdots & \ddots & \vdots & \vdots & \vdots \\ 0 & 0 & \dots &\bar B_{m-1} &  0& \bar\alpha^{(m-1)} \\\hline  0 & 0 & \dots & 0& I_{\ell-1} & \bar\alpha^{(m)}  \end{array}\right),$$
where $\bar B_1$ is any $\ell \times (\ell-1)$-submatrix of $B_1$, $\bar B_i$ is any  $\ell \times \ell$-submatrix of $B_i$ for $i=2,\dots, m-1$, and $\bar \alpha^{(j)} = ( \alpha^{(j)}_1, \dots,  \alpha^{(j)}_\ell)^\top$ for $j=1,\dots, m-1$, $\bar \alpha^{(m)} = ( \alpha^{(m)}_1, \dots,  \alpha^{(m)}_{\ell-1})^\top$,
is invertible, which is equivalent to 
$ \left(\begin{array}{c|c} \bar B_1 &  \bar\alpha^{(1)}  \end{array}\right)$
being invertible. This implies that $\hat B_1$ generates a $[\ell + r_1+1, \ell]$-MDS code. Analogously, $\hat B_2,\dots, \hat B_{m-1}$ generate MDS codes.

Since the last block forms a $[\ell + r_m, \ell]$- MDS code, we get that all maximal minors of 
\begin{equation*}
\left( \begin{array}{cccccccc} 
0&\dots &0&\alpha^{(m-1)}_\ell&  \alpha^{(m-1)}_\ell x^{(m)}_{\ell,1}&\dots&\alpha^{(m-1)}_\ell x^{(m)}_{\ell,r_m} \\ 
1&\dots&0&\alpha^{(m)}_1&  \alpha^{(m)}_1 x^{(m)}_{1,1}&\dots&\alpha^{(m)}_1 x^{(m)}_{1,r_m} \\ 
&\ddots&&\vdots&\vdots & & \vdots \\ 
0&\dots&1&\alpha^{(m)}_{\ell-1}&  \alpha^{(m)}_{\ell-1}x^{(m)}_{\ell-1,1}&\dots&\alpha^{(m)}_{\ell-1}x^{(m)}_{\ell-1,r_m}  
 \end{array}\right) 
\end{equation*}
are non-zero. Furthermore, all minors with the first $\ell$ columns of the first $m-1$ blocks and $\ell-1$ columns from the last block are non-zero, which implies that all maximal minors of $A$ are non-zero. 
It follows that 
\begin{equation}\label{eq:h}
\left( \begin{array}{cccccccc} 
1&0&\dots &0&\alpha^{(m-1)}_\ell&  \alpha^{(m-1)}_\ell x^{(m)}_{\ell,1}&\dots&\alpha^{(m-1)}_\ell x^{(m)}_{\ell,r_m} \\ 
0&1&\dots&0&\alpha^{(m)}_1&  \alpha^{(m)}_1 x^{(m)}_{1,1}&\dots&\alpha^{(m)}_1 x^{(m)}_{1,r_m} \\ 
\vdots&&\ddots&&\vdots&\vdots & & \vdots \\ 
0&0&\dots&1&\alpha^{(m)}_{\ell-1}&  \alpha^{(m)}_{\ell-1}x^{(m)}_{\ell-1,1}&\dots&\alpha^{(m)}_{\ell-1}x^{(m)}_{\ell-1,r_m}  
 \end{array}\right) 
\end{equation}
is the generator matrix of a $[\ell+r_m+1, \ell]$-MDS code.  By Lemma \ref{lem:equiv}, this last condition is equivalent to the condition that $\widehat{A}$ is the generator matrix of a $[\ell+r_m+1, \ell]$-MDS code.

The \emph{if}-direction can be shown analogously to the proof of Theorem~\ref{th:MRri}, again using Lemma \ref{lem:equiv} for the equivalence of the MDS property of \eqref{eq:h} and $\hat A$.
\end{proof}

We can now state the counterpart to Corollary \ref{co:MRr}, showing that one cannot construct PMDS codes with $s=1$
 over smaller fields.
\begin{cor}\label{co:MRr2}
Assuming that the MDS-conjecture (Conjecture \ref{MDSconj}) is correct, we have:
\begin{enumerate}
\item
If there exists a $[n,k=m\ell-1, \ell; r_1,\dots,r_m]$-PMDS code over $\F_q$ such that $\ell\in \{3,2^h-1\}$ and $\max_i \{r_i\} +\ell = 2^h+1$ (for some $h>1$), then $q\geq 2^h= \max_i\{r_i\}+\ell -1$.
\item
If there exists a $[n,k=m\ell-1, \ell; r_1,\dots,r_m]$-PMDS code over $\F_q$ such that the parameters are not included in Case 1 and $\ell >1$, then $q\geq \max_i\{r_i\}+\ell $.
\end{enumerate}
\end{cor}


\section{Decoding of PMDS Codes with $s=1$}\label{sec:decoding}

In this section we investigate decoding of PMDS codes with $s=1$. We will first give a general decoding algorithm for any such code, based on solving a linear system of equations arising from the parity check matrix of the code. Then we will comment on the special case that the block MDS codes are Reed-Solomon codes.

In the case that we have only $s=1$ additional erasure, there are $m-1$ blocks that have at most as many erasures as the erasure correction capability of the block MDS code. Therefore, one can use any suitable decoding algorithm for the block MDS code and decode each of these $m-1$ blocks separately. In the last block, which contains one erasure more than correctable by the code of the block, we need to use the additional parity from the PMDS property. With this we get the following result:

\begin{thm}
Let $C$ be $[n,k=m\ell-1, \ell; r_1,\dots,r_m]$-PMDS code over $\F_q$ and let $r\in \F_q^n$ be a received word that is decodable in $C$. Then the original codeword can be recovered from $r$ with a complexity of $O(m \max_i\{r_i\}^3)$ operations over $\F_q$.
\end{thm}
\begin{proof}
By Theorem \ref{th:classMRri} we know that $C$ has a generator matrix of the form \eqref{eq:Stand} and therefore a parity check matrix of the form
$$
 H= \left(\begin{array}{cccc|c} B_1^\perp & 0 & \dots & 0 & 0 \\  0 & B_2^\perp & \dots & 0 & 0 \\ 
\vdots & \vdots & \ddots & \vdots & \vdots \\ 0 & 0 & \dots & B_{m-1}^\perp & 0 \\\hline  X_1 & X_2 & \dots & X_{m-1}& A^\perp  \end{array}\right),
$$
where $B_i^\perp \in F_q^{r_i\times (r_i+\ell)}, A^\perp \in \F_q^{(r_m+1)\times(r_m+\ell)}$ denote parity check matrices of the codes generated by $B_i, A$, respectively, and
$$ X_i =  \left(\begin{array}{cccc} 0 & 0 & \dots & 0 \\ \vdots &\vdots &&\vdots\\ 0 & 0 & \dots & 0 \\ x_{i1}& x_{i2}&\dots&x_{i(r_i+\ell)} \end{array}\right)    \in \F_q^{(r_m+1)\times (r_i+\ell)} $$
is such that $M_i\; (A^\perp)^\top = - B_i \;X_i^\top$. Note that such $A,x_{i1},\dots,x_{i(r_i+\ell)}$ exist, since $M_i$ has rank $1$ and $B_i$ has full column space.

Each of the first $m-1$ blocks with at most $r_i$ erasures can be decoded by solving the linear system of equations arising from the matrix $B_i^\perp$\footnote{This system of equations has $r_i$ equations and at most $r_i$ variables corresponding to the erasures.}, which can be done with a complexity of order $O(r_i^3)$, using Gaussian elimination. Analogously, we can decode $r_m$ erasures in the last block with the first $r_m$ rows of $A^\perp$.

After correcting all blocks with at most $r_i$ erasures, we can decode the remaining block. 
If this block is one of the first $m-1$ blocks, we solve the system of equations arising from $B_i$ and the last row of $H$; if it is the last block we solve the system of equations arising from the last $r_m+1$ rows of $H$. These systems of equations have one variable and one equation more than in the previous case. Hence, the complexity order for solving this is still in $O(r_i^3)$. Since we have $m$ blocks, the statement follows.
\end{proof}

Note that in the proof of the previous theorem we simply used the parity check matrix for decoding erasures. As mentioned before one can also use other suitable erasure decoding algorithms in each block with at most $r_i$ erasures. This might be more efficient from a time or storage complexity, as well for the question how to store the code. In any case, the extra parity equation corresponding to the global parity needs to be stored and used for decoding the block with the extra erasure. 


For example, if every block MDS code is a Reed-Solomon code, we can use suitable algorithms for decoding any of the blocks with at most $r_i$ erasures, and Gaussian elimination for the block with the additional erasure, using the global parity equation. In that case we get an overall decoding complexity of order $O(m f_{RS}(\ell,\max_i\{r_i\})+\max_i\{r_i\}^3)$, where $f_{RS}(k,n-k)$ denotes the complexity of erasure decoding a received word in a $[n,k]$-Reed-Solomon code. Many algorithms have been developed for Reed-Solomon erasure decoding. E.g., the classical Berlekamp-Massey algorithm \cite{be68b,ma69} gives $ f_{RS}(\ell,\max_i\{r_i\})\in O(\max_i\{r_i\}^2)$; or, if the base field has characteristic $2$, the algorithm of \cite{li14} achieves  $ f_{RS}(\ell,\max_i\{r_i\})\in O(\ell+\max_i\{r_i\}\lg(\ell+\max_i\{r_i\}))$.


\section{Conclusions}\label{sec:conclusions}

In this paper we generalized the definition of PMDS codes to allow locality blocks of various length. We showed that this definition still fulfills the MR (maximally recoverable) property for codes with separate locality groups. Moreover, we gave a generalized construction for PMDS code with one global parity ($s=1$) by giving a generator matrix in systematic form. Then we showed that basically all PMDS codes of this type must have a generator matrix of this form. Based on the correctness of the MDS conjecture, we derived a necessary and sufficient field size for these codes to exist. The main result states that PMDS codes with $s=1$ exist if and only if the field size is at least the length of the longest locality block, except for a few special cases. For the few special cases our generalized construction gives codes over smaller fields than any other known construction. In the end we gave a simple decoding algorithm and derived its complexity order. 

A natural idea how to extend this work is to generalize these results to larger values of $s$. However, the number of conditions on the blocks of the generator matrix becomes quite large quite quickly. Thus, the techniques of this paper cannot straight-forwardly be transferred to larger values of $s$. However, for general $s$, some considerations can be found in the appendix of this paper.

\bibliographystyle{plain}
\bibliography{PMDS_stuff}


\section*{Appendix: Some Considerations for General~$s$}

The paper at hand characterizes all PMDS codes that can correct $r_i$ erasures locally in every block, plus $s=1$ additional erasures anywhere. A natural question that arises at this point is how to generalize the argument in order to get a characterization for a general $s>1$. Unfortunately the arguments for $s=1$ seem to be hard to be generalized. In fact, even a generalization of Lemma \ref{lem:Stand} to the case $s=2$ gives quite complicated conditions on the structure of a systematic generator matrix.

However, the main idea of Theorem \ref{th:classMRri}  can be generalized to any $s$. With this we get the following necessary conditions for the existence of $[n,k,\ell; r_1,\dots,r_m]$-PMDS codes for general $k$, respectively $s=m\ell-k$:

\begin{thm}\label{thm:almostfinal}

A $[n,k=m\ell -s, \ell; r_1,\dots,r_m]$-PMDS code over $\F_q$ exists only if there exist
\begin{itemize}
\item  a $[\ell+\max_i\{r_i\}+s,\ell]$-MDS code over $\F_q$ and
\item a $[m\ell , m\ell -s]$-MDS code over $\F_q$.
\end{itemize}
\end{thm}

\begin{proof}
The second statement, i.e., the existence of an  $[m\ell , m\ell -s]$-MDS code, easily follows from the definition of PMDS, since after puncturing each block in $r_i$ ($1\leq i\leq m$) coordinates, the remaining code must be a $[m\ell , m\ell -s]$-MDS code.

For simplicity we prove the first statement for $m=2$. The proof for larger $m$ is analogous. 
Let $C$ be a $[n,2\ell -s, \ell; r_1,r_2]$-PMDS code over $\F_q$. Similarly to the proof of Lemma \ref{lem:Stand} the  generator matrix in systematic form of $C$ is
$$ G=\left(\begin{array}{cc|cc} I_\ell & B & 0 & M \\\hline  0 & 0& I_{\ell-s} & A   \end{array} \right) ,$$
where $B\in \F_q^{\ell\times r_1}, A  \in \F_q^{(\ell-s)\times (r_2+s)}, M  \in \F_q^{\ell\times (r_2+s)}$. By the definition of PMDS codes, any $\ell - t$ columns from the first block and $\ell - s + t$ columns from the second block (for $0\leq t \leq s$) form an MDS code. We now consider the case $t=s$. In particular, if $M_s$ denotes the first $s$ columns of $M$, and we consider the minors including the first $\ell$ columns of the second block, all maximal minors of $( I_\ell \mid  B \mid M_s)$ are non-zero. This implies that $( I_\ell \mid B \mid M_s)$ generates a $[\ell+r_1+s,\ell]$-MDS code.

By symmetry the same holds for the second block. For this note that another generator matrix of $C$ has the following form:
$$ G'=\left(\begin{array}{cc|cc} A' & I_{\ell-s} & 0 & 0 \\\hline  M' & 0& I_{\ell} & B'   \end{array} \right) ,$$
where $B'\in \F_q^{\ell\times r_2}, A'  \in \F_q^{(\ell-s)\times (r_1+s)}, M  \in \F_q^{\ell\times (r_1+s)}$. With the same argument as above we get that  $(M'_s \mid  I_\ell \mid B )$ generates a $[\ell+r_2+s,\ell]$-MDS code.

\end{proof}

\begin{rem}
We saw in Theorem \ref{th:classMRri} that the two conditions given in Theorem \ref{thm:almostfinal} are also sufficient when $s=1$. Moreover, in that case the second condition is trivially satisfied, since  a $[m\ell, m\ell-1]$-MDS code exists over any finite field $\F_q$.
\end{rem}

\begin{cor}\label{cor:23}
Let $\ell, m,s > 1$. 
Assuming that the MDS Conjecture (Conjecture \ref{MDSconj})  is correct, this implies that for $q<\max\{\ell+\max_i\{r_i\}+s,   m\ell \}-2 $ no $[n,k=m\ell -s, \ell; r_1,\dots,r_m]$-PMDS codes exist.

Furthermore,
\begin{enumerate}
\item if $n^*:= \ell+\max_i\{r_i\}+s \leq   m\ell$ and $n^* \neq 2^h+2$ or $\ell \not\in \{3,2^h-1\}$, then no such PMDS code exists for $q<n^* -1$;
\item  if $n^*:=    m\ell\leq \ell+\max_i\{r_i\}+s $ and $n^* \neq 2^h+2$ or $m\ell -s \not\in \{3,2^h-1\}$, then no such PMDS code exists for $q<n^* -1$.
\end{enumerate}
\end{cor}

For $\ell=1$ however, we can always construct $[n, m-s, 1; r_1,\dots,r_m]$-PMDS codes as the concatenation of some repetition codes with a $[m,n-s]$-MDS code. This gives the following result:

\begin{thm}\label{thm:final}
For any $m\geq 2$, $s,r_1,\dots,r_m\geq 1$ there exists a $[n, m-s, 1; r_1,\dots,r_m]$-PMDS code  over a field $\F_q$ with $q\geq m-1$. If $m = 2^h+2$ (for some $h\in \mathbb N$) and $m-s\in \{3,2^h-1\}$, then such a code exists over a field of size  $q\geq m-1$.
\end{thm}

\begin{proof}
Let $G=(g_{i,j})_{i,j}$ be the generator matrix of a $[m,m-s]$-MDS code, which exists over $\F_q$ with $q\geq m-1$. Then replace all elements $g_{i,j}$ with the row vector $g_{i,j}(1,\dots,1)$ of length $r_j+1$. One can easily check that the resulting matrix is the generator matrix of a $[n, m-s, 1; r_1,\dots,r_m]$-PMDS code over $\F_q$.

In the case that $m = 2^h+2$ and $m-s\in \{3,2^h-1\}$ we can take as the outer $[m,m-s]$-MDS code a doubly extended Reed-Solomon code (or its dual), which exists over $\F_q$ with $q\geq m-2$. Then the same construction as above will result in a generator matrix of a PMDS code of the desired parameters.
\end{proof}

Note that in the setting of Theorem \ref{thm:final} the first necessary conditions of Theorem \ref{thm:almostfinal} becomes trivial, as it requires the existence of a one dimensional MDS code. The second condition is equivalent to the one of Theorem \ref{thm:final}, and hence also a sufficient condition for the existence of such a PMDS code.

We conclude this appendix with a few words on the easiest non-trivial case, i.e., when $s=2$. Even in this case, as we will see in the next example, it is very difficult to deduce a general construction and a characterization like the one we gave in this paper for $s=1$.

\begin{ex}
 Let $n=8, k=4$. Then the following is a generator matrix for a $[8,4,3; 1,1]$-PMDS code over $\F_7$:
$$\left( \begin{array}{cccc|cccc}
    1&0&0&1  &0&1&2&2\\
    0&1&0&4  &0&1&3&6\\
    0&0&1&6  &0&1&4&3\\
    0&0&0&0  &1&1&5&1
   \end{array}\right)
 $$
One can easily see that any combination of $3$ columns from each block together generates a $[6,4]$-MDS code over $\F_7$.

On the other hand, there is no possible completion of the matrix 
$$\left( \begin{array}{cccc|cccc}
    1&0&0&*  &0&1&1&*\\
    0&1&0&*  &0&1&2&*\\
    0&0&1&*  &0&1&3&*\\
    0&0&0&0  &1&1&4&*
   \end{array}\right)
 $$
such that it generates a $[8,4,3; 1,1]$-PMDS code over $\F_7$, although all properties stated in Theorem \ref{thm:almostfinal} are fulfilled. In particular, the respective submatrices as in the proof of Theorem \ref{thm:almostfinal} generate MDS codes, e.g., the first $3$ columns from both blocks together constitute the generator matrix of a $[6,4]$-MDS code, and the submatrix indexed by columns $1,2,3,4,6,7$ and rows $1,2,3$ can be completed to generate a $[6,3]$-MDS code.
\end{ex}

\end{document}